\documentclass[conference]{IEEEtran}
\usepackage{cite}
\usepackage{amsmath,amssymb,amsfonts}
\usepackage{amsthm}
\usepackage{algorithmic}
\usepackage{graphicx}
\usepackage{textcomp}
\usepackage{xcolor}
\def\BibTeX{{\rm B\kern-.05em{\sc i\kern-.025em b}\kern-.08em
    T\kern-.1667em\lower.7ex\hbox{E}\kern-.125emX}}

\makeatletter
\newcommand{\linebreakand}{%
  \end{@IEEEauthorhalign}
  \hfill\mbox{}\par
  \mbox{}\hfill\begin{@IEEEauthorhalign}
}
\makeatother

\allowdisplaybreaks
    
\begin{document}
\title{Bounds and Algorithms for Frameproof 
Codes and Related Combinatorial Structures}

\author{\IEEEauthorblockN{Marco Dalai\IEEEauthorrefmark{1}, Stefano Della Fiore\IEEEauthorrefmark{1}\IEEEauthorrefmark{2}, Adele A.  Rescigno\IEEEauthorrefmark{2} and Ugo Vaccaro\IEEEauthorrefmark{2}} \IEEEauthorblockA{\IEEEauthorrefmark{1}Department of Information Engineering, 
University {of Brescia}, Brescia (BS), Italy\\
Emails: \{marco.dalai, s.dellafiore001\}@unibs.it} \IEEEauthorblockA{\IEEEauthorrefmark{2}Department of Computer Science, University {of Salerno}, Fisciano (SA), Italy\\
Emails: \{sdellafiore, arescigno, uvaccaro\}@unisa.it}}

\newtheorem{thm}{Theorem}[section]
\newtheorem{lem}[thm]{Lemma}
\newtheorem{cor}[thm]{Corollary}
\newtheorem{prop}[thm]{Proposition}
\newtheorem{conj}[thm]{Conjecture}
\newtheorem{prob}[thm]{Problem}
\newtheorem{ques}[thm]{Question}
\newtheorem{ex}[thm]{Exercise}
\newtheorem{rem}[thm]{Remark}

\theoremstyle{definition} 
\newtheorem{defn}[thm]{Definition} 
\newtheorem{exa}[thm]{Example}

\newcommand{\bc}{{\bf  c}}
\newcommand{\remove}[1]{}

\maketitle

\begin{abstract}
In this paper, we study upper bounds on the minimum length of frameproof codes introduced by Boneh and Shaw \cite{Boneh}  to protect copyrighted materials. A $q$-ary $(k,n)$-frameproof code of length $t$ is a $t \times n$ matrix having  entries in  $\{0,1,\ldots, q-1\}$ and with the property  that for any column $\mathbf{c}$ and any other $k$ columns, there exists a row where the symbols of the $k$ columns are \emph{all}
different from the corresponding symbol 
(in the same row) of the column $\mathbf{c}$. 
In this paper, we show   the existence of $q$-ary $(k,n)$-frameproof codes of length $t = O(\frac{k^2}{q} \log n)$ for $q \leq k$, using the Lov\'asz Local Lemma, and of length $t = O(\frac{k}{\log(q/k)}\log(n/k))$ for $q > k$
using the expurgation method. Remarkably, for the practical  case of $q \leq k$ our findings give codes whose length almost matches the lower bound $\Omega(\frac{k^2}{q\log k} \log n)$ on the length of \emph{any} $q$-ary $(k,n)$-frameproof code and, more importantly, allow us to derive an algorithm of complexity $O(t n^2)$ for the construction of such codes.
\end{abstract}


\section{Introduction}


Frameproof codes were proposed  by Boneh and Shaw \cite{Boneh} as a tool to protect copyrighted materials from 
unauthorized use. The rationale is the following. When a distributor wants to sell $n$ copies of a digital 
product, she first chooses $t$ fixed locations in the digital product. Successively, for each product copy, she 
associates to each chosen location a $q$-ary symbol. Such a collection of chosen  symbols and locations in each 
product copy is known as a \emph{fingerprint}, which can be seen as a codeword of length $t$ over an alphabet of
size $q$. The clients do not know the locations and symbols stored in the data, which means they cannot remove
or modify them. However,  a group of malicious clients could  collude, share and compare their copies. In this 
scenario, the malicious clients  could in principle recover the locations and symbols of the fingerprints in order to create an 
illegal copy, that is, a copy of the digital product whose fingerprint equals that of an innocent client outside
the malicious coalition. To prevent this situation, frameproof codes can be used. 
Informally, a set of fingerprints (codewords)  is called a  $(k,n)$-frameproof code if 
any coalition of at most $k$ clients cannot create any other codeword outside their coalition.
Essentially, this means that any codeword possesses a piece of information (in some coordinate) not included in the same coordinate of any collection of at most $k$ other codewords (a more formal definition is provided  later on).

The main problems in the area consist of both providing bounds on the lengths of frameproof codes 
(this is motivated by the practical consideration that the length of the fingerprint to insert
in the digital data represents 
an obvious  overhead that needs to be minimized) \emph{and} computationally efficient procedures to 
construct frameproof codes of length close to the theoretical optimum. The seminal 
paper \cite{Boneh} spurred an interesting line of research to investigate the two problems mentioned above. Due to space constraints, it
is impossible here to summarize the many results in the area, and we refer the reader to the papers
\cite{Barg,Blackburn,Chee,Cheng,Guo,Shann2,Staddon,Stinson1998,Stinson,Stinson2,vanTrung} and references therein 
quoted for the relevant literature in the area. 
\subsection{Our Results} In this paper we show
 the existence of $q$-ary $(k,n)$-frameproof codes of length $t = O(\frac{k^2}{q} \log n)$ 
 for $q \leq k$, using the Lov\'asz Local Lemma, and of length $t = O(\frac{k}{\log(q/k)}\log(n/k))$ 
 for $q > k$,  using the expurgation method. We refer the reader to \cite{FL} and \cite{HLLT} for similar applications of these methods which are used to provide probabilistic constructions of separating codes and identifiable parent property (IPP) codes.
 We will discuss how our results 
 improve on known results present in the literature  after having  proved them. 
 Moreover,    for the   practically important  case of $q \leq k$
 (as motivated in \cite{Shann2}) we provide a $O(t n^2)$ randomized 
 algorithm  to construct codes  of length  $t = O(\frac{k^2}{q} \log n)$, that  almost matches the lower bound $\Omega(\frac{k^2}{q\log k} \log n)$ on the length of \emph{any} $q$-ary $(k,n)$-frameproof code. 
 To the best of our knowledge, this is the first polynomial time, in the code size $n$, algorithm with such a performance. We note that to explicitly construct $(k,n)$-frameproof codes of length $t$ any algorithm requires $\Omega(tn)$ time. 


\section{Preliminaries}
Throughout the paper, the logarithms without subscripts are in base two, and we denote with $\ln (\cdot)$ the natural logarithm. Given integers $a<b$,  we denote with $[a,  b]$ the set ${\{a, a+1, \ldots, b\}}$.
We start by introducing  the  combinatorial objects we study  in  this paper.
\begin{defn}\label{def:frameproof}
Let $k$, $n$, $q\geq 2$ be positive integers,  $n > k$.
A $q$-ary $(k, n)$-\emph{frameproof} code is a $t \times n$ matrix $M$ with entries in  $[0, q-1]$ such that for any column $\mathbf{c}$ and any other $k$ columns of $M$ we have that, there exists a row ${i \in [1,t]}$ where the symbols of the $k$ columns are all different from the symbol in the $i$-th row of the column $\mathbf{c}$. The number of rows $t$ of $M$ is called the length of the $q$-ary $(k,n)$-frameproof code.
\end{defn}

In order to provide upper bounds on the minimum length of $q$-ary $(k,n)$-frameproof codes we need to recall a 
strictly related class of codes, named $q$-ary $(k,n)$-strongly selective codes, 
studied in \cite{Vaccaro} by De Bonis and Vaccaro.

\begin{defn}\label{def:selective}
Let $k$, $n$, and $q\geq 2 $ be positive integers,  $n \geq k$.
A $q$-ary $(k,n)$-\emph{strongly selective} code is a $t \times n$ matrix $M$ with entries in $[0,q-1]$ such that for  any $k$-tuple of the columns of $M$ and  for any column $\mathbf{c}$ of the given $k$-tuple, there exists a row ${i \in [1,t]}$ such that $\mathbf{c}$ has an entry  $s\in [1,q-1]$ in row $i$  whereas the entries in row $i$ of all the remaining $k-1$ columns of the $k$-tuple belong to $[0, q-1]\setminus\{s\}$. The number of rows $t$ of $M$ is called the length of the $q$-ary $(k,n)$-strongly selective~code.
\end{defn}
Given an integer $q\geq 2$ and a $q$-ary vector $\bc\in [0, q-1]^t$, we denote with $w(\bc)$ the number of nonzero components of~$\bc$.

\begin{defn}\label{def:selectiveW}
A $q$-ary $(k,w,n)$-\emph{strongly selective} code is a $q$-ary $(k,n)$-strongly selective code with the additional constraint that each column $\mathbf{c}$ has $w(\mathbf{c}) = w$.
\end{defn}
There is a strong relation between $q$-ary $(k,n)$-frameproof codes and $q$-ary $(k+1, n)$-strongly selective codes. In fact, if we denote the minimum length of $q$-ary $(k,n)$-frameproof codes by $t_{\text{FP}}(q,k,n)$ and that of $q$-ary $(k, n)$-strongly selective codes by $t_{\text{SS}}(q,k,n)$, then the following lemma holds.

\begin{lem}\label{lem:SSFP}
$$
    t_{\text{SS}}(q,k+1,n) \geq t_{\text{FP}}(q,k,n) \geq \frac{1}{2} t_{\text{SS}}(q,k+1,n) \,.
$$
\end{lem}
\begin{proof}
The upper bound on $t_{\text{FP}}(q,k,n)$ easily follows by noticing that a $q$-ary $(k+1,n)$-strongly selective code of length $t$ is a $q$-ary $(k,n)$-frameproof code of length $t$ by definition of such codes.  Conversely, a $q$-ary $(k,n)$-frameproof code $\mathcal{C}$ is not necessarily a $q$-ary $(k+1,n)$-strongly selective code since we can have a column $\mathbf{c}$ and other $k$ columns in $\mathcal{C}$ where the only rows there exists are those in which the symbols of the $k$ columns do not contain the symbol of the column $\mathbf{c}$ and where this symbol is equal to $0$. Therefore, we are satisfying the property for frameproof codes but not the one for strongly selective codes. However, the minimum length $t_{\text{SS}}(q,k+1,n)$ of $q$-ary $(k+1,n)$-strongly selective codes is at most twice the minimum length $t_{\text{FP}}(q,k,n)$ of $q$-ary $(k,n)$-frameproof codes. Indeed, if one is given a $q$-ary $(k,n)$-frameproof code $\mathcal{C}$ then one can build a $(k+1,n)$-strongly selective code twice as long as $\mathcal{C}$ by taking the union of the rows of $\mathcal{C}$ and those of the complementary code $\overline{\mathcal{C}}$ obtained by replacing each entry $s$ in $\mathcal{C}$ by $q-1-s$. Hence we obtain the lower bound on $t_{\text{FP}}(q,k,n)$ shown in the statement of the lemma.
\end{proof}
Thanks to Lemma \ref{lem:SSFP} we can concentrate on finding bounds on the minimum length $t_{\text{SS}}(q,k,n)$ and then obtain, indirectly, bounds on $t_{\text{FP}}(q,k,n)$.

\section{A New Randomized Algorithm for frameproof codes via Lov\'asz Local Lemma}\label{sec:qsuperimposed}
In this section, we provide new upper bounds on the minimum length $t_{\text{SS}}(q,k,n)$ of 
$q$-ary $(k,n)$-strongly selective codes. We will provide two different bounds, one that is derived using the Lov\'asz Local Lemma while the other one using the expurgation method.

In \cite{Vaccaro}, it has been proved the existence of $q$-ary $(k,n)$-strongly selective of length
\begin{equation}\label{eq:BoundSS}
 t = O\left(\frac{k^2}{v} \log(n/k)\right),
\end{equation}
where $v = q-1$ for $q \leq k$ and $v = k$ for $q > k$. Their probabilistic contruction produces a randomized algorithm of complexity $\Theta(n^k)$ to generate $q$-ary $(k,n)$-strongly selective codes of length of the same order as the one shown in equation~\eqref{eq:BoundSS}. The algorithm  is
clearly impractical already for small  values of $k$.
Here we provide new bounds on $t_{\text{SS}}(q,k,n)$ that are of the same order as the ones provided by \cite{Vaccaro} but,
crucially, this is accompanied with a randomized construction algorithm 
of complexity $O(t n^2)$. This means that we can construct a $q$-ary $(k,n)$-strongly selective codes of length as in (\ref{eq:BoundSS}) in time 
polynomial in $n$ and $k$.

The key idea is to use the following matrices (introduced in \cite{KS}) where the constraints involve only pairs of columns.

\begin{defn}\label{lambda-matrix}
Let $n,w,\lambda$ be positive integers. 
A $q$-ary $t\times n$  matrix $M$, with entries in $[0, q-1]$,  
is a $(\lambda, w, n)$-matrix if the following properties hold true:\\
1) each column $\bc$ of $M$ has $w(\bc)=w$, \\
2)  any pair of columns $\mathbf{c}, \mathbf{d}$ of $M$ have at most $\lambda$ nonzero symbols
in common, that is, there are at most $\lambda$ rows 
where 
columns $\mathbf{c}$ and $\mathbf{d}$ have both the same entry $s\in [1, q-1]$.
\end{defn}

These matrices are related to $q$-ary $(k,w,n)$-strongly selective codes 
by the following result.
\begin{lem}\label{LaToSup}
If $M$ is a $q$-ary $t\times n$ $(\lambda, w, n)$-matrix with 
$\lambda=\left\lfloor (w-1)/(k-1)\right\rfloor$
then $M$ is a $q$-ary $(k,w,n)$-strongly selective code of length $t$.
\end{lem}
\begin{proof}
Let $\mathbf{c}$ be an arbitrary column of $M$ and let $A$ be the set of row indices in which column $\mathbf{c}$ has nonzero elements. Therefore $|A| = w$. Let $K$ be a set of arbitrary $k-1$ columns of $M$ where $\mathbf{c} \not \in K$. Let us denote by $M(A,K)$ the $w \times (k-1)$ submatrix of $M$ constructed by first selecting the $k-1$ columns in $K$ and then selecting the $w$ rows whose indices belong to $A$. Since $M$ is a $(\lambda, w, n)$-matrix we have that the number of nonzero elements that column $\mathbf{c}$ share (in the same rows) with any column in $K$ is at most $\lambda = \left\lfloor \frac{w-1}{k-1} \right\rfloor$. Hence, the total number of nonzero elements in $M(A,K)$ is at most $$(k-1) \lambda \leq (k-1) \left\lfloor \frac{w-1}{k-1} \right\rfloor \leq w-1\,.$$ Considering that $M(A,K)$ has $|A| = w$ we have that at least one row in $M(A,K)$ contains only zero elements. Then the lemma follows.
\end{proof}

Now, we want to give a good upper bound on the minimum length of $q$-ary $(\lambda,w,n)$-matrices that will provide us an upper bound on the minimum length of $q$-ary $(k,w,n)$-strongly selective codes.
We first need to recall the following well-known facts on binomial coefficients, where for positive integers $c \leq b \leq a$ we have
\begin{equation}\label{eq:ULBinom}
    \left(\frac{a}{b}\right)^b \leq \binom{a}{b} \leq \frac{a^b}{b!} \leq \left(\frac{ea}{b}\right)^b\,,
\end{equation}
\begin{equation}\label{eq:idBinom}
    \binom{a}{b} \binom{b}{c} = \binom{a}{c} \binom{a-c}{b-c}\,,
\end{equation}
and the following useful inequality, proved in \cite{Vaccaro1}, for positive integers $c \leq a \leq b$
\begin{equation}\label{eq:ineqBinom}
    \binom{a}{c} \Big/ \binom{b}{c} \leq \left(\frac{a - \frac{c-1}{2}}{b - \frac{c-1}{2}}\right)^c\,.
\end{equation}

Our main tool is the celebrated algorithmic version of the Lov\'asz Local Lemma for the symmetric case as given in \cite{MT}.

\begin{lem}\cite{MT} \label{lem:LLL}
Let $E_{1},E_{2},\ldots, E_{m}$ be events in a  probability space, 
where  each event $E_{i}$ is mutually independent of the set of all the other events 
$E_{j}$ except for at most $D$, and  $\Pr(E_{i})\leq P$ for all $1\le i\leq m$.
If $eP D\leq1$, then $\Pr(\cap_{i=1}^{m}\overline{E_{i}})>0$. 
Moreover, a configuration avoiding all events $E_i$ can be effectively found by using an average number 
of resampling at  most $m/D$.
\end{lem}

We are now ready to state our main lemma.

\begin{lem}\label{lem:lambdaMatrices}
There exists a $q$-ary $t \times n$ $(\lambda, w, n)$-matrix with
\begin{align}
t &= \max\Bigg\{\left\lceil 2w - (\lambda +1) \right\rceil,  \nonumber \\
&\left \lceil  \frac{\lambda}{2} +
 \frac{1}{q-1}\left (\frac{ew}{\lambda+1}\left(w-\frac{\lambda}{2} \right) (e(2n-4))^{\frac{1}{\lambda+1}}
\right) \right \rceil\Bigg\}. \label{boundont}
\end{align}
\end{lem}
\begin{proof}
Let $M$ be a random $t \times n$ $q$-ary matrix, $t \geq 2w - (\lambda +1)$, where each column $\mathbf{c}$ is picked uniformly at random among the set of all distinct $q$-ary vectors $\bc$ of length $t$ such that $w(\bc)=w$. It is easy to see that the number of such vectors is equal to $\binom{t}{w} (q-1)^w$, and therefore 
$$\Pr(\bc) = \left(\binom{t}{w} (q-1)^w\right)^{-1}\,.$$
Let $i,j\in [1, n], i\neq j$ and let us consider the event $\overline{E}_{i,j}$ that there exists \emph{at most} $\lambda$ rows  
such that both the $i$-th column and the $j$-th column of $M$ have the same 
nonzero symbol in \emph{each} of these rows.
We  evaluate the probability of the complementary ``\emph{bad}'' event ${E}_{i,j}$. Hence $E_{i,j}$ is the event that the $i$-th and $j$-th columns $\bc_i$ and $\bc_j$ have identical non-zero elements in at least $\lambda+1$ coordinates. We bound $\Pr(E_{i,j})$ by conditioning on the event that $\bc_i$ is equal to $c$.

For a subset $S\subset [1,t]$ of coordinates, let $E_{i,j}^S$ be the event that in each coordinate of $S$ the $i$-th and $j$-th column have identical non-zero elements. Finally, let $A$ be the set of coordinates where $\bc_i$ is non-zero.
Note that for $S\in\binom{A}{\lambda+1}$, i.e. $S$ is a subset of $A$ of size $\lambda +1$, we have
$$
\Pr(E_{i,j}^S|\bc_i = c)=\frac{\binom{t-(\lambda+1)}{w-(\lambda+1)} (q-1)^{w-(\lambda+1)}}{\binom{t}{w} (q-1)^w}\,.
$$
Then
\begin{align}
\Pr(E_{i,j}|\bc_i = c) & \leq \sum_{S\in \binom{A}{\lambda+1}}\Pr(E_{i,j}^S|\bc_i = c) \nonumber \\& = \binom{w}{\lambda + 1}\frac{{\binom{t-(\lambda+1)}{w-(\lambda+1)}} (q-1)^{w-(\lambda+1)}}{\binom{t}{w} (q-1)^w}.\label{eq: upperEij}
\end{align}
Since the right-hand side of \eqref{eq: upperEij} does not depend on the fixed column $c$, it also holds unconditionally. Hence
\begin{equation}\label{eq:bounPr}
\Pr(E_{i,j}) \leq \binom{w}{\lambda + 1}\frac{{\binom{t-(\lambda+1)}{w-(\lambda+1)}} (q-1)^{w-(\lambda+1)}}{\binom{t}{w} (q-1)^w}.
\end{equation}
Therefore, by \eqref{eq:bounPr} we have
\begin{align}
\Pr({E}_{i,j}) &\leq \binom{w}{\lambda + 1}{\binom{t-(\lambda+1)}{w-(\lambda+1)}} \Big/ \binom{t}{w} (q-1)^{\lambda +1} \nonumber \\
&\stackrel{(i)}{=} \binom{w}{\lambda + 1} \binom{w}{\lambda + 1} \Big/ \binom{t}{\lambda + 1} (q-1)^{\lambda +1} \nonumber \\
&\stackrel{(ii)}{\leq} \frac{1}{(q-1)^{\lambda +1}} \binom{w}{\lambda + 1} \left( \frac{w - \frac{\lambda}{2}}{t - \frac{\lambda}{2}} \right)^{\lambda + 1} \nonumber \\
&\stackrel{(iii)}{\leq}\frac{1}{(q-1)^{\lambda +1}} \left(\frac{ew}{\lambda + 1} \right)^{\lambda+1} \left( \frac{w - \frac{\lambda}{2}}{t - \frac{\lambda}{2}} \right)^{\lambda + 1} =P\,,
\label{Q-two}
\end{align}
where $(i)$ holds due to equality \eqref{eq:idBinom}, $(ii)$ is true due to inequality \eqref{eq:ineqBinom}, and finally $(iii)$ holds thanks to inequalities~\eqref{eq:ULBinom}.

The number of events ${E}_{i,j}$ is equal to $n(n-1)/2$.
Let us {fix} an event  $E_{i,j}$,  the number of events  from which   $E_{i,j}$ can be dependent is equal to $D=2n-4$. Hence, according to Lemma \ref{lem:LLL}, if  $P$  (as defined in \eqref{Q-two}) and $D=2n-4$ satisfy $e P D\leq 1$, then the probability that {\em none} of the ``bad"  events $E_{i,j}$ occurs is strictly  positive. One can see that  by setting  $t$ as in the second term of the maximum in \eqref{boundont} one indeed obtains
$$e P D=  \frac{e(2n-4)}{(q-1)^{\lambda+1}}\left (\frac{ew}{\lambda+1} \right )^{\lambda+1} \left (\frac{w-\frac{\lambda}{2}}{t-\frac{\lambda}{2}} \right )^{\lambda+1} \leq 1.$$
Hence, from Lemma \ref{lem:LLL}  one can construct a  $q$-ary $(\lambda, w, n)$-matrix $M$ whose number of rows $t$ satisfies equality~\eqref{boundont} since we also need to consider the initial assumption $t \geq 2w - (\lambda +1)$ so that all the computation carried out in this lemma are meaningful.
\end{proof}

Now, thanks to Lemmas \ref{LaToSup} and \ref{lem:lambdaMatrices}, we can prove the 
following result.

\begin{thm}\label{thm:algo}
    There exists a randomized algorithm to construct a $q$-ary $(k,w,n)$-strongly selective code with length
    \begin{align}
    t &\leq 1 + \max\Bigg\{2w - \frac{w-1}{k-1}, \:\: \frac{w-1}{2(k-1)} + \nonumber\\ &\frac{ew(k-1)}{(q-1)(w-1)}\left(w-\frac{w-1}{2(k-1)} +\frac{1}{2} \right) (e(2n-4))^{\frac{k-1}{w-1}}\Bigg\}. \label{eq:upperWstrong}
\end{align}
The algorithm requires, on average, time $O(t n^2)$ to construct the code.
\end{thm}
\begin{proof}
The upper bound on $t$ shown in the statement of the lemma is derived by substituting the value of $\lambda = \lfloor (w-1)/(k-1) \rfloor$ of Lemma \ref{LaToSup} into equation \eqref{boundont} of Lemma \ref{lem:lambdaMatrices}, and by using the inequalities $\frac{w-1}{k-1}-1 \leq \left\lfloor \frac{w-1}{k-1} \right\rfloor \leq \frac{w-1}{k-1}$.
The time complexity $O(t n^2)$ comes from Lemma \ref{lem:LLL} by noticing that $m/D = n(n-1)/(4n-8) \leq n/3$. Moreover, the algorithm that one obtains from \cite{MT} requires to randomly generate a matrix, checking if the $\Theta(n^2)$ events $\overline{E}_{i,j}$ are satisfied, and resample \emph{only on non-satisfied} events. This means that we need to check if the $i$-th column and the $j$-th column of the matrix have at most $\left\lfloor\frac{w-1}{k-1}\right\rfloor$ nonzero elements in common, that can be done with at most $O(t)$ operations, and resample only over non-satisfied events. Then we need to check only the events that involve columns that have been resampled. Altogether, by Lemma \ref{lem:LLL} this procedure requires  
$O(tn^2 + n \cdot m/D \cdot t ) = O(t n^2)$ elementary operations.
\end{proof}
\begin{rem}
Before optimizing (\ref{eq:upperWstrong}) over the parameter $w$, we would like to stress that Theorem 
\ref{thm:algo} is not a mere technical intermediate result, but might be important in 
several practical scenarios. For instance, it has been shown in \cite{Vaccaro} that
$q$-ary $(k,n)$-strongly selective codes can be used to solve important communication 
problems arising in multiple-channel wireless networks. In the scenario considered in \cite{Vaccaro}, one has a set of $n$ uncoordinated stations, attempting transmission 
over a set of $q-1$ independent 
channels. Transmission is successful if and only if no two stations attempt to
transmit over the same channel at the same time instant.
The idea  in \cite{Vaccaro} is the following:  Each station is assigned a distinct codeword 
of  a $q$-ary $(k,n)$-strongly selective code, and such a codeword corresponds 
to the  transmission
schedule
of the associated station.
The presence of a symbol $s\in [0,q-1]$ in the  $i$-th coordinate of a given 
codeword $\bc$ naturally translates as the "instruction", to the station possessing
codeword $\bc$, to stay "silent" in the $i$-th step of the communication protocol
if $s=0$, and to transmit over the $j$-th channel in the $i$-th step if $s=j\neq 0$.
Under the hypothesis that at any given time at most $k$ stations are "active", the authors of 
\cite{Vaccaro} proved that $q$-ary $(k,n)$-strongly selective codes naturally correspond to 
conflict resolution protocols in multiple-channel wireless networks. Now, in many situations,
it is important not only to minimize the length of the protocol (i.e., the number of time 
instants before all stations transmit successfully) but it is important also to restrict  the 
number of attempted transmissions by each station (to save energy, for example, see \cite{Vaccaro1} for more).  
Therefore $q$-ary $(k,w,n)$-strongly selective codes, where each codeword has $w$ non-zero components,
could be useful in these instances. We remark that the 
 techniques of \cite{Vaccaro} 
are not able to deal with this hard constraint on the number of attempted transmissions by each station, nor suggest efficient algorithms for the construction of $q$-ary $(k,n)$-strongly selective codes, as our technique does.

\end{rem}
We  optimize $w$ in equation \eqref{eq:upperWstrong} to get a randomized algorithm for (unconstrained) $q$-ary $(k,n)$-strongly selective codes.

\begin{cor}\label{cor:stronglysel}
    Let $w = \lceil 1 + (k-1) \ln (2en) \rceil$, then the algorithm described in Theorem \ref{thm:algo} constructs a $q$-ary $(k,w,n)$-strongly selective code that is, clearly, a $q$-ary  $(k,n)$-strongly selective code of length $t$ upper bounded as
    \begin{align*}
       t &\leq \max\Bigg\{2(k-1)\ln(2en) - \ln(n), \:\: \frac{\ln(n)}{2} + \nonumber\\ &\frac{e^2 (k-1)^2}{q-1} \ln(2en) + \frac{7e^2(k-1)}{2(q-1)} \Bigg\} + O(1).
    \end{align*}
\end{cor}
\begin{proof}
    Fixing $w$ as in the statement of the corollary and using the following inequalities
    \begin{equation*}
       1 + (k-1) \ln (2en) \leq \lceil 1 + (k-1) \ln (2en) \rceil \leq 2 + (k-1) \ln (2en)
    \end{equation*}
    by Theorem \ref{thm:algo}, considering only the second term of the maximum, we have that
    \begin{align*}
         t &\leq 1 + \frac{1 + (k-1) \ln(2en) }{2(k-1)} + \frac{e}{(q-1)\ln(2en)} \\
        & (2 + (k-1) \ln(2en)) \left(\frac{3}{2} + \left(k-1\right) \ln(2en)\right) (2en)^{\frac{1}{\ln(2en)}} \\
        &\leq \frac{\ln(n)}{2} + \frac{e^2 (k-1)^2}{q-1} \ln(2en) + \frac{7e^2(k-1)}{2(q-1)} + O(1)\,,
    \end{align*}
    since $k \geq 2$, $\ln(2en) \geq 2$ and $(2en)^{\frac{1}{\ln(2en)}} = e$. Then, the corollary follows since we also need to consider the first term of the maximum of Theorem \ref{thm:algo}.
\end{proof}

Therefore due to Lemma \ref{lem:SSFP} and Corollary \ref{cor:stronglysel} we obtain the main result of this paper.

\begin{thm}\label{thm:algoframeproof}
There exists a randomized algorithm to construct a $q$-ary $(k,n)$-frameproof code with length
\begin{align*}
  t \leq \max&\Bigg\{2k\ln(2en) - \ln(n), \:\: \frac{\ln(n)}{2} + \nonumber\\ &\frac{e^2 k^2}{q-1} \ln(2en) + \frac{7e^2k}{2(q-1)} \Bigg\} + O(1).
\end{align*}
The algorithm requires, on average, time $O(t n^2)$ to construct the code.
\end{thm}

To properly judge the value of Theorem \ref{thm:algoframeproof}, we recall the following result
that provides a lower bound on the length of any $q$-ary $(k,n)$-frameproof code.

\begin{thm}\label{thm:lb}\cite{DBV,Vaccaro, Shann2}
Given positive integers $q,k$, and $n$, with $q\geq 2$ and $2\leq k\leq \sqrt{n}$, the minimum length 
of any $q$-ary $(k,n)$-frameproof code satisfies
\begin{equation}\label{eq:lb}
t_{\text{FP}}(q,k,n)=\Omega\left(\frac{k^2}{q\log k}\log\frac{n}{k}\right).
\end{equation}
\end{thm}
Therefore, one can see that the construction method provided by Theorem \ref{thm:algoframeproof},
besides being quite efficient, produces codes of almost optimal length.

We can also prove the following result.
%

\begin{thm}\label{thm:lowerByShann}
Let $M$ be a $q$-ary $(k,n)$-frameproof code of \emph{minimum} length $t$. Then
\begin{equation}\label{eq:almost}
   \left\lceil\frac{n}{q-1}\right\rceil \geq t \geq \left\lceil \frac{1}{q} \min\left\{n, \frac{15+\sqrt{33}}{24}k^2\right\} \right\rceil\,.
\end{equation}
\end{thm}
\begin{proof}
We can define a map from $M$ to a binary $(k,n)$-strongly selective code of length $t'$ by mapping each symbol $i \in [0,q-1]$ into $\mathbf{e}_{i+1}$, where $\mathbf{e}_i$ is the binary column vector that has $1$ in the $i$-th component and $0$ elsewhere. Therefore $t' = q t$. Now, the right-hand side inequality of (\ref{eq:almost})
follows since ${t' \geq \min\left\{n, \frac{15+\sqrt{33}}{24}k^2\right\}}$ by \cite[Theorem 2]{Shann}.
To prove the left-hand side inequality of (\ref{eq:almost}), simply observe that by
taking any $n$ columns from the side-by-side
concatenation of 
$(q-1)$ many $\lceil n/(q-1)\rceil\times\lceil n/(q-1)\rceil$ diagonal matrices, where the $i$-th matrix has symbol
$i\in[1,q-1]$ on its diagonal, and 0 elsewhere, one gets  a $q$-ary $(k,n)$-frameproof code
of length $\lceil n/(q-1)\rceil$.
\end{proof}

\remove{
Now, we are going to compare our results with the following upper bound on $t_{\text{FP}}(q,k,n)$ given in \cite{Stinson2} by Stinson, Wei and Chen.

\begin{thm}\label{thm:StinsonFrameproof}
There exists a $q$-ary $(k,n)$-frameproof code of length
$$
    t \leq  - k \ln\left(n \frac{k!}{k!-1}\right) \Big/ \ln\left(1-\left(1-\frac{1}{q}\right)^{k}\right)\,.
$$
\end{thm}
We note that in the original paper \cite{Stinson2}, this bound is stated in terms of bounds on the 
length of separating hash families. However, it is well known  that a separating hash family of type $(t,n,q,\{1, k\})$ 
is equivalent to a $(k,n)$-frameproof code of length $t$. 

In the following theorem, we prove that our bound given in Theorem \ref{thm:algoframeproof} provides a better result than the one of Theorem \ref{thm:StinsonFrameproof} whenever $k$ is sufficiently large and $q$ is sufficiently small with respect to $k$.

\begin{thm}\label{thm:compOurStinson}
For $k$ sufficiently large and $q \leq 0.3118 \cdot k (1 + o(1))$, where $o(1)$ is meant for $k \to \infty$, the bound given in Theorem \ref{thm:algoframeproof} is better than the one of Theorem \ref{thm:StinsonFrameproof}.
\end{thm}
\begin{proof}
The bound given in Theorem \ref{thm:StinsonFrameproof} can be lower bounded as follows
\begin{align*}
    &- (1+o(1)) k \ln n \Big/ \ln\left(1-\left(1-\frac{1}{q}\right)^{k}\right) \geq \\ 
    & \qquad - (1+o(1))k \ln n \Big/ \ln\left(1-\left(1-e^{-k/q}\right)\right) \geq
    \\
    & \hspace{4cm} (1+o(1)) \left(e^{k/q} - 1\right) k \ln n\,,
\end{align*}
since $(1+x) \leq e^{x}$ and $\ln(1+x) \geq \frac{x}{1+x}$ for all $x > -1$. Then, for $q \leq k$, the bound given in Theorem \ref{thm:algoframeproof} can be stated as follows
$$
    (1+o(1)) \frac{e^2 k^2}{q-1} \ln(n)\,.
$$
Let us take $q = \alpha k (1+o(1))$, where $\alpha$ is a real in $(0,1]$, then we need to study the following inequality to compare the two bounds
$$
    e^2 \leq \alpha \left(e^{\frac{1}{\alpha}}-1\right)\,.
$$
Since the right-hand side of this inequality is a strictly decreasing function in $\alpha$ for $\alpha > 0$,  we can compute numerically the range of $\alpha$ for which the inequality is satisfied, that is, $(0, 0.3118]$.
\end{proof}

We remark that, in \cite{Shann2}, Shangguan, Wang, Ge and Miao obtain the following upper bound on the minimum length of $(k,n)$-frameproof codes.

\begin{thm}\label{thm:ShannFrameproof}
If $q \leq k$,  then there exists a $q$-ary $(k,n)$-frameproof code of length
$$
 t \leq \frac{-k \ln n  - (k+1) \ln 2}{\ln\left[1 - \left(1-\frac{q-1}{k+1}\right) \left(\frac{q-1}{k+1}\right)^k - \frac{q-1}{k+1} \left(1-\frac{1}{k+1}\right)^k\right]}\,.
$$
\end{thm}

\begin{rem}
We note that the bound in Theorem \ref{thm:ShannFrameproof} improves the ones provided in Theorems \ref{thm:algoframeproof} and \ref{thm:StinsonFrameproof} when $q \leq k$.  However,  neither the procedure used to prove Theorem \ref{thm:StinsonFrameproof} nor the one used in Theorem \ref{thm:ShannFrameproof} suggests a randomized algorithm of time complexity $O(t n^2)$ to construct the code.  They only achieve a complexity of order $O(n^k)$. 
\end{rem}

}
\section{New improved upper bounds on $t_{\text{FP}}(q,k,n)$ via Expurgation Method}
In this Section, we will provide an existential upper bound on the minimum 
length of $q$-ary $(k,n)$-frameproof codes that improves the best results known 
in the literature. 
We first recall  such known results.
\begin{thm}\label{thm:StinsonFrameproof}\cite{Stinson2}
There exists a $q$-ary $(k,n)$-frameproof code of length
$$
    t \leq  - k \ln\left(n \frac{k!}{k!-1}\right) \Big/ \ln\left(1-\left(1-\frac{1}{q}\right)^{k}\right)\,.
$$
\end{thm}
We note that in the original paper \cite{Stinson2}, this bound is stated in terms of bounds on the 
length of separating hash families. However, it is well known  that a separating hash family of type $(t,n,q,\{1, k\})$ 
is equivalent to a $q$-ary $(k,n)$-frameproof code of length $t$. In \cite{Shann2}, the authors provide the following bound that improves  the one of Theorem \ref{thm:StinsonFrameproof} whenever $q \leq k$ and $k$ is sufficiently large.
\begin{thm}\label{thm:ShannFrameproof}\cite{Shann2}
If $q \leq k$,  then there exists a $q$-ary $(k,n)$-frameproof code of length
$$
 t \leq \frac{-k \ln n  - (k+1) \ln 2}{\ln\left[1 - \left(1-\frac{q-1}{k+1}\right) \left(\frac{q-1}{k+1}\right)^k - \frac{q-1}{k+1} \left(1-\frac{1}{k+1}\right)^k\right]}\,.
$$
\end{thm}

Our result employs 
the so-called expurgation method or deletion method, 
the same
technique of \cite{Stinson2},  but with a more careful optimization of the parameters.
Therefore, we are able to improve the bound of Theorem \ref{thm:StinsonFrameproof} for every value of the parameters $k$, $q$, and $n$. Here we state the following theorem.

\begin{thm}\label{thm:exp}
There exists a $(k,n)$-frameproof code of length $t$,  where $t$ is the minimum integer such that the following inequality holds
\begin{equation*}
(k+1) \binom{n \left( 1 + \frac{1}{k} \right)}{k} (1-p_{q, k})^t  \leq 1\,,
\end{equation*}
where
$$
	p_{q,k} := \begin{cases} \left(1-\frac{1}{q}\right)^k &\hspace{-0.2cm}\text{for } q > k, \\ \left(1-\frac{q-1}{k+1}\right) \left(\frac{q-1}{k+1}\right)^k + \frac{q-1}{k+1} \left(1-\frac{1}{k+1}\right)^k &\hspace{-0.2cm}\text{otherwise.} \end{cases}
$$
\end{thm}

\begin{proof}
Let $M$ be a $q$-ary $t \times (n+\ell)$ matrix, where each element is picked i.i.d.  at random in the set $[0,q-1]$ with distribution $\mathbf{\mu}= (\mu_0, \mu_1, \ldots, \mu_{q-1})$ that will be fixed later. For a given index $i \in [1, n]$ and a set of column-indices $B$, $|B| = k-1$, $i \not \in B$,  let $E_{i, B}$ be the event such that for every row in which $\mathbf{c}_i$ (the $i$-th column) has a symbol $s$,  there exists an index $j \in B$ such that $\mathbf{c}_j$ has symbol $s$ in that same row. Therefore the probability of each event $E_{i,B}$ can be upper bounded as 
$\Pr(E_{i,B}) \leq \left( 1 - \sum_{i=0}^{q-1} \mu_i \left(1-\mu_i\right)^k \right)^t$.
The number of such events is equal to $(n+\ell) \binom{n+\ell-1}{k}$. Now, let $X$ be the random variable that represents the number of events $E_{i,B}$ that are satisfied. Hence,  taking $\mu_i = 1/q$ for every $i \in [0,q-1]$ (uniform distribution) when $q > k$ and $\mu_i = 1/(k+1)$ for every $i\in[1,q-1]$, $\mu_0 = 1-(q-1)/(k+1)$ when $q \leq k$, we obtain
\begin{equation*}
	\mathbb{E}[X] \leq (n+\ell) \binom{n+\ell-1}{k}  \left(1-p_{q,k} \right)^t.
\end{equation*}
We note that if $\mathbb{E}[X] < \ell+1$ then there exists at most $\ell$ ``\emph{bad}'' events $E_{i,B}$ that are satisfied. Then, for each of these events $E_{i,B}$ we remove one column with index in $\{i\} \cup B$. Hence we are left with a $q$-ary matrix with $t$ rows and at least $n$ columns that satisfy the frameproof property. Therefore we obtain a $(k,n)$-frameproof code with length $t$. Thus the theorem follows taking $\ell = \lfloor n/k \rfloor$.
\end{proof}

\begin{cor}\label{cor:exp}
Using the inequalities in equation \eqref{eq:ULBinom}, we have that, from Theorem \ref{thm:exp}, the length of $(k,n)$-frameproof codes is upper bounded as follows
$$
	t \leq \frac{- k \ln\left( n \frac{k+1}{k} \right) - \ln\left( \frac{k+1}{k!} \right)}{\ln \left(1-p_{q,k} \right)}\,,
$$
where $p_{q,k}$ is the same quantity defined in Theorem \ref{thm:exp}.
\end{cor}

\begin{thm}\label{thm:compOurStinson2}
The bound of Corollary \ref{cor:exp} improves the one of Theorem \ref{thm:StinsonFrameproof} for every $n \geq k \geq 2$ and $q >k$.
\end{thm}
\begin{proof}
We need to prove the following inequality
\begin{equation}\label{eq:comp}
\frac{- k \ln\left( n \frac{k+1}{k} \right) - \ln\left( \frac{k+1}{k!} \right)}{\ln \left(1-p_{q,k} \right)} < \frac{- k \ln\left(n \frac{k!}{k!-1}\right)}{\ln\left(1-\left(1-\frac{1}{q}\right)^{k}\right)}\,.
\end{equation}
Since for $q > k$, $p_{q,k} = \left(1-1/q\right)^k$, we can rearrange and simplify the terms in \eqref{eq:comp} to obtain the following inequality.
\begin{equation}\label{eq:compOS}
\left(\frac{k+1}{k}\right)^k \frac{k+1}{k!} < \left(\frac{k!}{k!-1}\right)^k\,.
\end{equation}
Now, since $\left(1+\frac{1}{k}\right)^k \leq e$ and $k! \geq e \left(\frac{k}{e}\right)^k$, the left-hand side of \eqref{eq:compOS} is smaller than $\left(\frac{e}{k}\right)^k (k+1)$.
To prove inequality \eqref{eq:compOS}, it suffices to show that
\begin{equation}\label{eq:compOS2}
     \left(\frac{e}{k}\right)^k (k+1) < 1\,,
\end{equation}
since the right-hand side of \eqref{eq:compOS} is greater than $1$ for every $k$.

It can be seen that the left-hand side of \eqref{eq:compOS2} is a decreasing function in $k$ for $k \geq 2$. The first integer $k$ for which inequality \eqref{eq:compOS2} holds is $k = 5$. Then, the theorem follows since inequality \eqref{eq:compOS} also holds for ${k=2,3,4}$ by direct computation.
\end{proof}

\begin{thm}\label{thm:compOurShann}
The bound of Corollary \ref{cor:exp} improves the one of Theorem \ref{thm:ShannFrameproof} for every $n \geq k \geq 2$ and $q \leq k$.
\end{thm}
\begin{proof}
Clearly, we need to prove the following inequality
\begin{multline}\label{eq:comp1Shann}
\frac{- k \ln\left( n \frac{k+1}{k} \right) - \ln\left( \frac{k+1}{k!} \right)}{\ln \left(1-p_{q,k} \right)}
\\
 < \frac{-k \ln n  - (k+1) \ln 2}{\ln\left[1 - \left(1-\frac{q-1}{k+1}\right) \left(\frac{q-1}{k+1}\right)^k - \frac{q-1}{k+1} \left(1-\frac{1}{k+1}\right)^k\right]}\,.
\end{multline}
For $q \leq k$, the denominators in \eqref{eq:comp1Shann} are equal by definition of $p_{q,k}$. Therefore, we can rearrange and simplify the terms to obtain the following inequality.
\begin{equation}\label{eq:comp2Shann}
\left(\frac{k+1}{k}\right)^k \frac{k+1}{k!} < 2^{k+1}\,.
\end{equation}
Proceeding as in the proof of Theorem \ref{thm:compOurStinson2}, the left-hand side of \eqref{eq:comp2Shann} is a decreasing function in $k$ for $k \geq 2$ while the right-hand side of \eqref{eq:comp2Shann} is increasing in $k$. By inspection, it is easy to see that inequality \eqref{eq:comp2Shann} holds even for $k=2$.
\end{proof}

\section*{Acknowledgements}
The work of A. A. Rescigno and U. Vaccaro was partially supported by project SERICS (PE00000014) under the NRRP MUR program funded by the EU--NGEU.

\end{document}